\documentclass[sigconf,nonacm]{acmart}

\usepackage{amsmath}

\usepackage{amssymb}
\usepackage{algorithm}
\usepackage{algpseudocode}
\usepackage{listings}

\renewcommand\footnotetextcopyrightpermission[1]{}
\title{Need for Speed Sort: A Recursive Distribution-Based Sorting Algorithm}

\author{Fran Sučić}
\affiliation{
  \institution{Independent Researcher}
  \country{Croatia}}
\email{fransucic46@gmail.com}

\author{Leo Vitasović}
\affiliation{
  \institution{IT University of Copenhagen}
  \country{Denmark}}
\email{leov@itu.dk}

\author{Nikola Petrušić}
\affiliation{
  \institution{Independent Researcher}
  \country{Croatia}}
\email{nikola.petrusic221@gmail.com}

\begin{document}

\begin{abstract}
We present \emph{Need for Speed Sort} (NFS Sort), a recursive distribution-based sorting algorithm designed for numeric arrays. The algorithm partitions elements into equal-width value intervals, recursively refines dense buckets, and propagates analytical interval bounds between recursive calls, avoiding repeated scans for local minima and maxima. NFS Sort combines a fragment-based, cache-conscious scatter procedure for large subarrays with a lower-overhead auxiliary-array approach for smaller inputs. Small buckets are deferred to a final insertion-sort cleanup, while a comparison-based fallback is activated when recursive partitioning repeatedly fails to reduce the problem size. This mechanism guarantees a worst-case running time of $\mathcal{O}(n \log n)$ and auxiliary space usage of $\mathcal{O}(\log n)$. Experimental evaluation on synthetic inputs and real-world datasets from the SOSD benchmark suite compares NFS Sort with Balanced Learned Sort, IPS4o, Boost Spreadsort, PDQSort, and \texttt{std::sort}. The results show that NFS Sort is competitive or better than established state-of-the-art sorting methods across dataset sizes and distributions, outperforming the learned baseline particularly on smaller inputs while retaining strong performance at larger scales. Overall, NFS Sort combines efficient recursive distribution, practical memory management, and robust worst-case guarantees for high-performance numeric sorting.
\end{abstract}

\maketitle


\section{Introduction}

Sorting is a fundamental problem in computer science and a core component of
many software systems. Comparison-based algorithms such as quicksort,
mergesort, heapsort, and introsort are robust and widely used, but they are
limited by the $\Omega(n \log n)$ lower bound in the comparison model~\cite{cormen2009introduction}. 
For numeric values, distribution-based methods can avoid this bound by
exploiting the structure of the value range~\cite{cormen2009introduction,knuth1998sorting}.

We introduce \emph{Need for Speed Sort} (NFS Sort), a recursive distribution-based
sorting algorithm for numeric arrays. At each recursive step, NFS Sort divides
the current value range into equally spaced buckets, scatters elements into
those buckets, and recursively processes only those that remain larger than the set threshold. Small
buckets are left for a final insertion-sort pass, which is efficient because
all remaining disorder is confined to bounded-size regions.

NFS Sort combines ideas from existing distribution sorters. Inspired by Flashsort~\cite{neubert1998flashsort},
it maps elements to buckets according to their relative position in the observed
value range. Unlike a single-pass approach, it recursively refines dense
regions. Its implementation also uses memory-management techniques inspired by
IPS$^4$o~\cite{axtmann2017ips4o}.

To avoid repeated memory accesses, NFS Sort propagates
analytical bucket bounds instead of recomputing minima and maxima during
recursion. It also includes a fallback mechanism: if recursion repeatedly fails
to reduce the size of a bucket, that bucket is sorted using a comparison-based
sort. This preserves a worst-case $\mathcal{O}(n \log n)$ bound while retaining
the practical advantages of distribution sorting.

This paper describes the algorithm, analyses its time and space complexity,
discusses adversarial inputs, and evaluates its empirical performance against
established sorting methods.

To facilitate reproducibility, the source code for NFS Sort and the
accompanying benchmark implementation are publicly available at
\url{https://github.com/fsucic/need-for-speed-sort}.

The remainder of the paper is organized as follows. Section~2 discusses related
work. Section~3 describes the algorithm and its implementation details.
Section~4 analyzes the theoretical behavior of NFS Sort, including worst-case
complexity, expected recursion depth under distributional assumptions, and
adversarial inputs. Section~5 presents experimental results, and Section~6 concludes.

\section{Related Work}

Sorting algorithms are commonly divided into comparison-based and
distribution-based methods. NFS Sort belongs to the distribution-based family.
Distribution-based sorting algorithms exploit properties of the keys rather
than relying only on comparisons. Bucket sort, radix sort, and counting sort can
achieve linear time under suitable assumptions, but their performance depends on
the key domain, input distribution, and memory requirements~\cite{cormen2009introduction,knuth1998sorting}. Flashsort is particularly relevant to this work because it maps elements to buckets based on
their relative position between the minimum and maximum input values, followed
by local clean-up~\cite{neubert1998flashsort}.

More recent work has focused on engineering distribution sorting for modern
hardware. IPS$^4$o~\cite{axtmann2017ips4o}, for example, uses efficient classification, 
buffering, and block permutation techniques to improve cache behaviour and reduce memory
overhead. NFS Sort is inspired by these ideas, but applies a sequential
recursive distribution strategy: it repeatedly refines dense value intervals,
propagates analytical bounds between recursive calls, and uses a fallback sorter
to avoid pathological recursion.

\section{Algorithm}

This section describes the design of NFS Sort. We first present a naive
version of the algorithm to illustrate its core principles. We then
introduce several optimisations that are incorporated in the final
implementation and noticeably improve performance.

NFS Sort does not introduce a fundamentally new sorting paradigm, but
rather combines existing ideas into an efficient overall design.
Conceptually, it can be viewed as a recursive variant of Flashsort~\cite{neubert1998flashsort}
augmented with modern memory-management techniques inspired by
IPS$^4$o~\cite{axtmann2017ips4o}.

\subsection{Naive NFS Sort}

Appendix~\ref{app:naive-nfs} gives code-style pseudocode for the naive version of NFS Sort.

We assume an array of numeric elements stored contiguously in memory.
The algorithm begins by computing the minimum and maximum values of the
input array, which requires $\mathcal{O}(n)$ time. These values, together
with the array boundaries, are then passed to the recursive core
procedure.

At each recursive step, the algorithm chooses a number of buckets based on
the size of the current subarray. Our implementation uses a naive
piecewise heuristic and never uses more than 6500 buckets.
The algorithm then performs a \emph{scatter step}, which
distributes the elements of the subarray into buckets.

First, the global value range $[\textit{min}, \textit{max}]$ is divided into
equally sized buckets, determined by
\[
    \textit{interval\_size}
    =
    \frac{\textit{max} - \textit{min}}{\textit{bucket\_count} - 1}.
\]

\noindent
Each element $x$ is then assigned to a bucket according to its relative
position in the range $[\textit{min}, \textit{max}]$. Specifically, the bucket
index is computed as
\[
    b(x)
    =
    \left\lfloor
    \frac{(x - \textit{min})(\textit{bucket\_count} - 1)}
    {\textit{max} - \textit{min}}
    \right\rfloor .
\]

\noindent
Equivalently, if
\[
    \textit{inverse\_interval}
    =
    \frac{1}{\textit{max} - \textit{min}},
\]
then the classification rule can be written as
\[
    b(x)
    =
    \left\lfloor
    (x - \textit{min})(\textit{bucket\_count} - 1)
    \cdot \textit{inverse\_interval}
    \right\rfloor .
\]

\noindent
This mapping sends $\textit{min}$ to the first bucket and $\textit{max}$ to
the last bucket.\footnote{With this particular choice of scaling by
$\textit{bucket\_count}-1$, the last bucket contains only values equal to
$\textit{max}$, assuming exact arithmetic. This convention is intentional and
simplifies the handling of the maximum element.}
Values between $\textit{min}$ and $\textit{max}$ are assigned
proportionally, so that the bucket boundaries are evenly spaced across the
full value range.

Elements are rearranged such that all elements assigned to the same
bucket occupy a contiguous region of the array, while the buckets
themselves appear in sorted order with respect to their value ranges.
This distribution step, which we refer to as the scatter step, runs in linear
time in the size of the current subarray.

Intuitively, this operation partitions the input according to value
ranges, thereby decomposing the problem into smaller independent
subproblems that can be processed recursively.

During the scatter step, the algorithm additionally records the number
of elements assigned to each interval. The subsequent processing of a
bucket depends on its size. Buckets containing at most a small threshold
number of elements are left unchanged and handled later by a final
cleanup phase. Larger buckets are processed recursively.

Importantly, recursion does not recompute the minimum and maximum values
of a bucket. Instead, analytical bounds are propagated. Since the
original value range is divided into equally sized intervals, the
corresponding interval boundaries directly provide valid lower and upper
bounds for each recursive call. These analytically derived bounds avoid
additional linear scans while preserving correctness.

The propagated bounds are interval boundaries, not necessarily the actual
minimum and maximum values contained in the bucket. A bucket may contain elements
that occupy only part of its assigned interval, but all of its elements lie
within the bounds passed to the recursive call.

Extremely skewed distributions must nevertheless be considered. If a large fraction
of the elements is assigned to a single bucket (for instance, in distributions containing a large amount of extremely small values), the recursive reduction
may fail to sufficiently decrease the problem size. To address this, the algorithm employs a fallback mechanism. If a recursive
branch repeatedly produces a bucket containing more than $\alpha$
fraction of its parent subarray, the algorithm allows for only a small fixed
number of such consecutive failures before falling back to a
comparison-based sort on that bucket.

Assuming the fallback sorter has worst-case $\mathcal{O}(m \log m)$ time
on a bucket of size $m$, this mechanism prevents indefinitely unbalanced
recursion and yields an overall worst-case bound of
$\mathcal{O}(n \log n)$, as shown in Section~\ref{sec:analysis}.

After the recursive phase terminates, the algorithm performs insertion
sort on the entire array. At this point, all remaining unsorted regions
are buckets whose sizes are bounded by a fixed threshold, and the buckets
themselves are already ordered by value range. Therefore, all remaining
inversions are confined within constant-size regions. For a fixed
threshold, the total number of remaining inversions is $\mathcal{O}(n)$,
so the final insertion-sort pass runs in $\mathcal{O}(n)$ time.

The naive version described here deliberately omits several engineering
details and edge-case refinements in order to emphasize the main recursive
bucketing structure. These refinements, including degenerate ranges,
helper arrays, and the scatter step, are introduced in the
following subsections.

\subsection{Implementation Details and Refinements}

This section describes several implementation details and refinements that
improve the practical performance of the algorithm.

\subsubsection{Data Types}

NFS Sort is implemented for arithmetic numeric types, including standard
integral and floating-point types. The current implementation assumes finite
inputs and does not define behaviour for special floating-point values such as
\texttt{NaN}, \texttt{+inf}, or \texttt{-inf}.

\subsubsection{Fast Paths}

NFS Sort includes fast paths for already sorted and reverse-sorted inputs. When
the top-level sorting procedure is called, the implementation first performs
quick-fail checks to determine whether the array is already in increasing or
decreasing order. If the array is already sorted, the algorithm returns it
immediately as-is. If the array is reverse sorted, the array is reversed and then
returned.

If neither fast path applies, then the input is neither sorted nor reverse
sorted. For finite numeric inputs, this also implies that the array does not
consist entirely of equal elements. Hence, the minimum and maximum values are
distinct, and the algorithm can proceed.

\subsubsection{Approximating the Minimum and Maximum}

One optimization used by NFS Sort is to approximate the minimum and maximum
values before entering the recursive procedure. Instead of performing a full
linear scan over the input solely to determine these values, the implementation
estimates them by sampling elements from the array.

If the approximate minimum and maximum are equal, the sampled bounds are
discarded and the true minimum and maximum are computed by a full scan. This
prevents the initial call from receiving a degenerate value range.

The approximated bounds are then passed to the first recursive call. Since these
bounds may not contain all input values, the initial scatter step uses two
additional buckets: one for elements below the approximated minimum and one for
elements above the approximated maximum. During this step, the
algorithm also computes the true minimum and maximum. Whenever an element is
placed into the underflow bucket, it is compared against the current minimum;
similarly, elements placed into the overflow bucket are compared against the
current maximum.

Thus, after the first scatter step, the true minimum and maximum are known. These
exact bounds are then propagated to lower levels of recursion, so the overflow
and underflow buckets are needed only in the initial distribution step.

\subsubsection{Handling Homogeneous Subarrays}

At the beginning of each recursive call after the initial one, the algorithm
performs a quick-fail check whether the current subarray is homogeneous, meaning 
that all of its elements are equal. If so, the subarray is already sorted and the 
recursive call returns immediately.

This case occurs naturally on inputs with few distinct values. After several
distribution steps, many buckets may contain repeated copies of a single value.
The homogeneity check detects such buckets early and avoids unnecessary bucket
construction, scattering, and further recursion.

\subsubsection{Scatter Step}

For large subarrays, NFS Sort uses a fragment-based scatter procedure inspired
by IPS$^4$o~\cite{axtmann2017ips4o}. Instead of immediately writing each
element to its final bucket region, the algorithm first classifies elements into
buckets and stores them in small fixed-size buffers, called fragments. Each
bucket has its own fragment buffer. When a fragment becomes full, it is flushed
back into the input array as a contiguous block.

During this pass, the algorithm also counts the number of elements assigned to
each bucket. These counts are then prefix-summed to compute the final bucket
boundaries. Since full fragments were written sequentially during the first
pass, they are not necessarily located inside their final bucket regions. A
second phase therefore defragments the array by moving or swapping whole
fragments into their correct bucket intervals. Any elements that do not fill a
complete fragment are kept temporarily in the helper buffer and copied into the
remaining gaps at the end.

This approach is similar in spirit to IPS$^4$o~\cite{axtmann2017ips4o}: both algorithms classify
elements into buckets, buffer writes in small blocks, and move data in
cache-friendly fragments rather than as individual elements. This is done to
reduce scattered memory writes and improve locality during distribution.
However, the version used in NFS Sort is simpler and sequential. It uses the
fragment scatter only as the distribution step inside the recursive algorithm,
rather than as part of a fully in-place parallel sorting framework.

For smaller subarrays, NFS Sort uses a simpler scatter procedure. The algorithm
first classifies each element into a bucket and counts the resulting bucket
sizes. These counts are then converted into prefix sums, which determine the
contiguous output region assigned to each bucket.

After the bucket boundaries have been computed, the algorithm performs a second
linear pass over the subarray. Each element is classified again and written to
its corresponding position in a helper array using the bucket offsets as write
pointers. Finally, the helper array is copied back into the original subarray.

This naive scatter step performs more direct element-wise writes than the
fragment-based version, but it has lower overhead and is therefore preferable on
small inputs. The implementation also records whether enough sufficiently large
buckets remain to justify further recursive processing; otherwise, the
subarray is left for the final cleanup phase. \footnote{All buckets left for the final cleanup phase are always bounded in
size by the fixed threshold \(S\).}

\subsubsection{Auxiliary Storage}
The auxiliary storage described in this subsection consists of the fixed
work arrays allocated by the top-level procedure. The additional temporary
bucket-offset vectors created during recursion are discussed separately in
Section~\ref{sec:space-complexity}.

NFS Sort allocates a fixed set of auxiliary arrays once in the top-level
procedure and reuses them throughout all recursive calls. This avoids repeated
allocation inside the recursion and keeps the additional memory usage
predictable.

The implementation maintains an array of bucket sizes, used by both scatter
procedures to count how many elements are assigned to each bucket. For the
fragment-based scatter, it also maintains an array of fragment sizes, which
records how many elements are currently stored in each per-bucket fragment
buffer. In addition, the fragment-based scatter uses a helper buffer containing
one fixed-size fragment for each bucket, two integer arrays for bucket end
offsets and current bucket write offsets, and a temporary fragment buffer used
when swapping fragments.

For small subarrays, the naive scatter reuses the helper buffer as a temporary
output array: elements are written into the helper array according to their
bucket offsets and then copied back to the original subarray. In both scatter
variants, the auxiliary arrays are temporary workspaces only; the final sorted
output remains in the original input array.

With fragment size set to 90, the implementation allocates a fragment helper
array of size $2002 \cdot 90$ elements and a temporary fragment of size $90$.
Thus, the auxiliary storage for elements of type $T$ is
\[
    (2002 \cdot 90 + 90)\cdot \texttt{sizeof}(T)
    =
    180270 \cdot \texttt{sizeof}(T).
\]
\noindent
The implementation also allocates four integer work arrays: one array of 6500
bucket sizes and three arrays of 2002 offsets. This contributes
\[
    (6500 + 2002 + 2002+ 2002)\cdot \texttt{sizeof}(\texttt{int})
    =
    12506 \cdot \texttt{sizeof}(\texttt{int})
\]
\noindent
additional bytes. Therefore, excluding recursion-stack overhead and small
temporary vectors of bucket offsets, the total auxiliary memory is
\[
    180270 \cdot \texttt{sizeof}(T)
    +
    12506 \cdot \texttt{sizeof}(\texttt{int}).
\]

For example, when sorting \texttt{double} values and assuming 4-byte integers,
this is approximately $1.42$ MiB of auxiliary memory.

Note that all of the above values were chosen using simple heuristics and may be subject
to further optimization.

The memory usage of the small temporary vectors of bucket offsets is analyzed 
in Section~\ref{sec:space-complexity}.

\subsubsection{Floating-Point Arithmetic and Index Clamping}

For floating-point inputs, bucket indices are computed using finite-precision
arithmetic. Small rounding errors may occur near bucket boundaries, especially
when values lie very close to the boundary between two adjacent buckets.

The important property for correctness is that bucket regions remain ordered by
value. In exact arithmetic, this follows directly from the classification
formula: within a recursive call, every element is transformed using the same
affine map with positive scale. Thus, if \(x < y\), then the classification
value of \(x\) is no larger than the classification value of \(y\).

For finite floating-point inputs and non-degenerate intervals, the corresponding
computation is monotone in normal use. Rounding may affect which side of a
bucket boundary an element falls on, but it does not normally make bucket
regions appear in the wrong order. In other words, floating-point error may
slightly perturb boundary cases, but it should not cause smaller values to be
placed after larger values at the bucket level.

To guard against numerical errors that could produce an index just outside the
valid range, the implementation clamps each computed bucket index before using
it to access bucket-related arrays. In exact arithmetic, every value in the
interval \([\textit{min},\textit{max}]\) maps to an index between \(0\) and
\(\textit{bucket\_count}-1\). The clamping step is therefore only an
implementation safeguard against finite-precision arithmetic and does not
change the intended bucket assignment under exact arithmetic.

\section{Theoretical Analysis}
\label{sec:analysis}

This section analyzes the theoretical behavior of NFS Sort. We first show that
the fallback mechanism guarantees a worst-case running time of
$\mathcal{O}(n \log n)$. We then analyze the auxiliary space usage and show that,
because the implementation uses fixed-size helper arrays together with recursive
bucket-offset vectors, the total auxiliary space is $\mathcal{O}(\log n)$.

We also discuss the expected recursion depth under reasonably balanced input
distributions. In such cases, the recursive bucket refinement rapidly reduces
subproblem sizes, causing the algorithm to behave close to linearly in practice.
Finally, we examine adversarial inputs for the algorithm.

\subsection{Worst-Case Time Complexity}

We now prove that NFS Sort has worst-case running time
$\mathcal{O}(n \log n)$.

The proof relies on the following implementation conditions, all of which are
satisfied by the algorithm described above.

First, every scatter step on a subarray of size $m$ runs in
$\mathcal{O}(m)$ time. This holds because the number of buckets used by the
implementation is bounded by a fixed constant in the large-array cases and is
at most $m$ in the small-array case.

Second, the scatter step partitions the current subarray into contiguous
bucket regions ordered by their value intervals. Thus, if bucket $i$ precedes
bucket $j$, with $i<j$, then every element assigned to bucket $i$ is less than
or equal to every element assigned to bucket $j$. Consequently, after scattering, 
no inversion crosses bucket boundaries; any remaining disorder is internal to 
individual buckets.

Third, recursive calls are made only on buckets larger than a fixed cleanup
bound $S$. Buckets of size at most $S$ may be left for the final insertion-sort
cleanup. The constant $S$ is independent of the input size $n$.

Finally, the fallback sorter has worst-case running time
$\mathcal{O}(m \log m)$ on a subarray of size $m$.

Inputs containing \texttt{NaN} or
infinities are outside the scope of the theorem.

\begin{lemma}
NFS Sort has worst-case running time $\mathcal{O}(n \log n)$.
\end{lemma}

\begin{proof}
Consider a recursive call on a subarray of size $m$. Let $k$ be the number of
buckets used by the call. All non-recursive work performed by the call,
including homogeneity checks, bucket initialization, classification, prefix-sum
computation, scattering, and bookkeeping, consists of a constant number of
linear passes over the subarray together with $\mathcal{O}(k)$ bucket-related
work. Therefore the cost of one recursive call, excluding recursive calls made
on its children, is
\[
    \mathcal{O}(m+k).
\]
Since the implementation ensures $k=\mathcal{O}(m)$, this cost is
\[
    \mathcal{O}(m+k)=\mathcal{O}(m).
\]
After scattering, the subarray is partitioned into buckets of sizes
\[
    m_1,m_2,\ldots,m_k,
\]
\noindent
with

\[
    \sum_{i=1}^k m_i = m.
\]
Buckets of size at most $S$ are not processed recursively and are left for the
final cleanup phase. Buckets larger than $S$ are either processed recursively or
sorted by the fallback sorter.

Let $\alpha$ be the fixed constant used to detect insufficient progress, where
$0<\alpha<1$. A recursive child of size $m_i$ is called pathological if
\[
    m_i > \alpha m.
\]
\noindent
Otherwise it is non-pathological, and its size satisfies
\[
    m_i \leq \alpha m.
\]
\noindent
Therefore, every non-pathological recursive step reduces the subproblem size by
at least a fixed constant factor.

The algorithm permits at most a fixed number $c$ of consecutive pathological
recursive steps along any recursion branch. If more than $c$ consecutive
pathological steps would occur, the current bucket is sorted using the fallback
sorter instead of being processed recursively.

It follows that along any root-to-leaf recursion path, recursive calls can be
grouped into phases. Each phase consists of at most $c$ consecutive
pathological calls, followed either by fallback termination or by one
non-pathological call. Hence each phase has constant length, since $c$ is a
fixed implementation constant.

Every phase that does not terminate by fallback contains a non-pathological
call. At that point, the size of the continuing recursive subproblem is at most
an $\alpha$ fraction of the size at the beginning of the phase. Thus, after
each non-terminating phase, the subproblem size decreases by a fixed constant
factor. Such a decrease can occur only $\mathcal{O}(\log n)$ times before the
subproblem size falls below the fixed cleanup bound $S$. Since each phase has
bounded length, the total number of recursive scatter levels on any
root-to-leaf path is $\mathcal{O}(\log n)$.

Now consider all scatter work at a fixed recursion depth. At that depth, the
active subarrays are disjoint, because every scatter step partitions its parent
subarray into disjoint bucket regions. Hence the sum of the sizes of all active
subarrays at any fixed depth is at most $n$. Since the non-recursive cost of a 
call is linear in the size of its subarray, the total non-recursive work at any 
one depth is $\mathcal{O}(n)$.

There are only $\mathcal{O}(\log n)$ recursive depths, so the total
non-recursive work over all recursive calls is
\[
    \mathcal{O}(n \log n).
\]

It remains to bound the total cost of fallback sorting. Suppose the fallback
sorter is applied to subarrays of sizes
\[
    q_1,q_2,\ldots,q_r.
\]
\noindent
These fallback subarrays are pairwise disjoint: each one is a bucket in the
recursive partition tree, and once a bucket is fallback-sorted, recursion does
not continue inside it. Therefore,
\[
    \sum_{j=1}^r q_j \leq n.
\]
By assumption, fallback sorting a subarray of size $q_j$ costs
$\mathcal{O}(q_j \log q_j)$. Hence the total fallback cost is
\[
    \sum_{j=1}^r \mathcal{O}(q_j \log q_j).
\]
\noindent
Since $q_j \leq n$ for every $j$,
\[
    \sum_{j=1}^r q_j \log q_j
    \leq
    \sum_{j=1}^r q_j \log n
    =
    \left(\sum_{j=1}^r q_j\right)\log n
    \leq
    n\log n.
\]
\noindent
Thus the total fallback cost is
\[
    \mathcal{O}(n \log n).
\]

Finally, after recursive distribution and fallback sorting terminate, NFS Sort
performs a global insertion-sort cleanup pass. By construction, every region
that remains potentially unordered has size at most the fixed constant $S$.
Furthermore, scatter steps place bucket regions in sorted value-interval order,
so no inversion crosses the boundary between two final regions. Therefore all
remaining inversions are internal to these final regions. For any final unordered 
region $B$,
\[
    \binom{|B|}{2}
    \leq
    \frac{S|B|}{2},
\]
\noindent
because $|B|\leq S$. Since the final regions are disjoint, the total number of
remaining inversions $I$ satisfies
\[
    I
    \leq
    \sum_B \binom{|B|}{2}
    \leq
    \sum_B \frac{S|B|}{2}
    =
    \frac{S}{2}\sum_B |B|
    \leq
    \frac{Sn}{2}
    =
    \mathcal{O}(n).
\]
\noindent
Insertion sort runs in $\mathcal{O}(n+I)$ time on an array with $I$ inversions~\cite{cormen2009introduction,knuth1998sorting}.
Since $I=\mathcal{O}(n)$, the final cleanup phase runs in $\mathcal{O}(n)$ time.

Combining the non-recursive work over all recursive calls, fallback work, and
cleanup work gives
\[
    \mathcal{O}(n \log n)
    +
    \mathcal{O}(n \log n)
    +
    \mathcal{O}(n)
    =
    \mathcal{O}(n \log n).
\]

Therefore, NFS Sort has worst-case running time $\mathcal{O}(n \log n)$.
\end{proof}

\subsection{Space Complexity}
\label{sec:space-complexity}

Most of the auxiliary storage used by NFS Sort is allocated once in the
top-level procedure and reused throughout the recursion. This includes the
helper array, fragment buffers, temporary fragment storage, bucket-size arrays,
fragment-size arrays, and offset arrays used by the scatter procedures. Since
the maximum number of buckets and the fragment size are fixed implementation
constants, this part of the auxiliary memory usage is constant with respect to
the input size \(n\).

The only additional auxiliary storage that grows with recursion comes from the
temporary bucket-offset vectors created inside recursive calls. Each call to
the recursive core procedure stores a vector of bucket offsets returned by the
scatter step. If the call uses \(k\) buckets, then this vector contains
\(k+1\) offsets in the usual case. During the first scatter step, two additional
overflow buckets may be used to handle elements below and above the approximate
initial bounds, so the first call may store \(k+3\) offsets. Since the
implementation uses at most 6500 buckets, the number of offsets stored by any
single recursive call is bounded by a fixed constant.

Although many recursive calls may occur over the whole execution, their
bucket-offset vectors are not all live at the same time. At any moment, the
only live vectors are those belonging to recursive calls on the current
call stack. Once a recursive call returns, its local bucket-offset vector is
destroyed before control continues to other branches of the recursion tree.

Therefore, the memory used by bucket-offset vectors is proportional to the
maximum recursion depth. As shown in the worst-case time analysis, the fallback
mechanism prevents arbitrarily long chains of insufficiently shrinking
recursive calls. Along any root-to-leaf path, after at most a fixed number of
pathological steps, either the current bucket is sorted by the fallback sorter
or the subproblem size decreases by a fixed constant factor. Hence the maximum
recursion depth is \(\mathcal{O}(\log n)\).

Since each recursive level stores only a constant-size bucket-offset vector,
the total live memory used by bucket-offset vectors is
\[
    \mathcal{O}(\log n).
\]
More concretely, if \(K_{\max}=6500\) is the maximum number of buckets, then
each non-initial recursive call stores at most \(K_{\max}+1\) integer offsets,
and the initial call stores at most \(K_{\max}+3\) integer offsets. Thus, if
the maximum recursion depth is \(D\), the live storage used by bucket-offset
vectors is bounded by
\[
    (K_{\max}+3)D \cdot \texttt{sizeof}(\texttt{int}),
\]

up to lower-order constant factors. Since \(K_{\max}\) is fixed and
\(D=\mathcal{O}(\log n)\), this is asymptotically \(\mathcal{O}(\log n)\).

Combining this with the fixed auxiliary buffers described earlier, the total
auxiliary space usage of the implementation is

\[
    \mathcal{O}(1) + \mathcal{O}(\log n)
    =
    \mathcal{O}(\log n),
\]

where the \(\mathcal{O}(1)\) term hides the fixed-size helper arrays and
fragment buffers. In practical terms, the dominant auxiliary memory cost is the
fixed helper storage, while the recursive bucket-offset vectors contribute only
a small additional amount proportional to the recursion depth.

\subsection{Expected Recursion Depth}
\label{sec:expected-depth}

In this section, we provide intuition for why NFS Sort can perform well on
typical inputs. We first prove a simple theorem that relates the expected
recursion depth to the density of the input distribution. The theorem is not
intended to characterize every implementation detail of the algorithm; rather,
it serves as a guide for understanding why recursive bucket refinement can
quickly reduce subproblem sizes. We then apply the result to several concrete
input distributions.

We will give a simple bound on the expected recursion depth under a bounded
input density model. Suppose the input elements
\[
    X_1, X_2, \ldots, X_n
\]
are drawn independently from a probability density function \(f\) supported on
an interval \([a,b]\). Assume that \(f\) is bounded, and define
\[
    M := (b-a)\|f\|_{\infty}.
\]
The quantity \(M\) measures how far the density may deviate from the uniform
density on \([a,b]\). For example, if \(f\) is uniform on \([a,b]\), then
\(M=1\).

Assume that each recursive distribution step divides the current value interval
into \(B\) equal-width buckets, and that recursion stops on a branch once the
corresponding bucket contains at most \(T\) elements.

\begin{theorem}
Let
\[
    M := (b-a)\|f\|_{\infty}.
\]
For the branch containing \(X_1\), let \(D\) be the first depth \(d\) such that
the depth-\(d\) bucket containing \(X_1\) contains at most \(T\) sample points.
Then
\[
    \mathbb{E}[D]
    \leq
    \max\left\{0,\left\lceil \log_B\left(\frac{nM}{T}\right)\right\rceil\right\}
    +
    \frac{B}{B-1}.
\]
In particular,
\[
    \mathbb{E}[D] = \mathcal{O}(\log n).
\]
\end{theorem}

\begin{proof}
For each depth \(d\geq 0\), let \(I_d(X_1)\) be the unique depth-\(d\) bucket
that contains \(X_1\). Since each recursive step divides the current interval
into \(B\) equal-width buckets, the length of a depth-\(d\) bucket is
\[
    |I_d(X_1)| = \frac{b-a}{B^d}.
\]
Then define
\[
    N_d := \#\{1 \leq i \leq n : X_i \in I_d(X_1)\},
\]
the number of sample points lying in the same depth-\(d\) bucket as \(X_1\).
The stopping depth is therefore
\[
    D := \min\{d\geq 0 : N_d \leq T\}.
\]
Equivalently,
\[
    \{D>d\} = \{N_d > T\}.
\]
Using the tail-sum formula for nonnegative integer-valued random variables,
\[
    \mathbb{E}[D]
    =
    \sum_{d\geq 0}\Pr(D>d)
    =
    \sum_{d\geq 0}\Pr(N_d>T).
\]
We next bound \(\Pr(N_d>T)\). Fix \(d\), and condition on the value
\(X_1=x\). Then \(I_d(X_1)\) becomes the deterministic bucket \(I_d(x)\)
containing \(x\). Its probability mass is
\[
    p_d(x)
    :=
    \Pr(X_j \in I_d(x))
    =
    \int_{I_d(x)} f(u)\,du.
\]
Since \(f(u)\leq \|f\|_{\infty}\) for all \(u\), we have
\[
    p_d(x)
    \leq
    \|f\|_{\infty}|I_d(x)|
    =
    \|f\|_{\infty}\frac{b-a}{B^d}
    =
    \frac{M}{B^d}.
\]
Conditioned on \(X_1=x\), the remaining random variables
\(X_2,\ldots,X_n\) are independent of \(X_1\), and each falls in \(I_d(x)\)
with probability \(p_d(x)\). Therefore
\[
    N_d \mid (X_1=x)
    \stackrel{d}{=}
    1+\operatorname{Bin}(n-1,p_d(x)).
\]
Hence
\[
    \mathbb{E}[N_d-1\mid X_1=x]
    =
    (n-1)p_d(x)
    \leq
    \frac{(n-1)M}{B^d}
    \leq
    \frac{nM}{B^d}.
\]
Taking expectation over \(X_1\), we obtain
\[
    \mathbb{E}[N_d-1]
    \leq
    \frac{nM}{B^d}.
\]
Since \(N_d\) is integer-valued, the event \(N_d>T\) implies
\[
    N_d-1 \geq T.
\]
By Markov's inequality,
\[
    \Pr(N_d>T)
    =
    \Pr(N_d-1\geq T)
    \leq
    \frac{\mathbb{E}[N_d-1]}{T}
    \leq
    \frac{nM}{TB^d}.
\]
Since probabilities are also at most \(1\),
\[
    \Pr(D>d)
    =
    \Pr(N_d>T)
    \leq
    \min\left\{1,\frac{nM}{TB^d}\right\}.
\]
Let
\[
    A := \frac{nM}{T}.
\]
Then
\[
    \mathbb{E}[D]
    \leq
    \sum_{d\geq 0}\min\left\{1,\frac{A}{B^d}\right\}.
\]
Choose
\[
    d_0
    :=
    \max\left\{0,\left\lceil \log_B A \right\rceil\right\}.
\]
By definition of \(d_0\), we have \(A\leq B^{d_0}\). Therefore, for every
\(d\geq d_0\),
\[
    \frac{A}{B^d}
    \leq
    B^{d_0-d}
    \leq 1.
\]
Splitting the sum at \(d_0\), we get
\[
    \mathbb{E}[D]
    \leq
    \sum_{d=0}^{d_0-1} 1
    +
    \sum_{d=d_0}^{\infty} \frac{A}{B^d}.
\]
The first sum is \(d_0\), and the second satisfies
\[
    \sum_{d=d_0}^{\infty} \frac{A}{B^d}
    \leq
    \sum_{d=d_0}^{\infty} B^{d_0-d}
    =
    \sum_{k=0}^{\infty} B^{-k}
    =
    \frac{1}{1-B^{-1}}
    =
    \frac{B}{B-1}.
\]
Thus
\[
    \mathbb{E}[D]
    \leq
    d_0 + \frac{B}{B-1}.
\]
Substituting the definition of \(d_0\), we obtain
\[
    \mathbb{E}[D]
    \leq
    \max\left\{0,\left\lceil \log_B\left(\frac{nM}{T}\right)\right\rceil\right\}
    +
    \frac{B}{B-1}.
\]
Since \(M\), \(T\), and \(B\) are fixed constants independent of \(n\), this
implies
\[
    \mathbb{E}[D] = \mathcal{O}(\log n).
\]
\end{proof}

We now apply the theorem to several typical input distributions. These examples
are meant to illustrate the scale of the bound, rather than to precisely model
every possible input.

Assume that the input values lie in the interval
\[
    [a,b] = [0,10000],
\]
and suppose that each recursive distribution step uses
\[
    B = 2000
\]
buckets. We stop recursion once a bucket contains at most
\[
    T = 10
\]
elements, and we consider an input size of
\[
    n = 1{,}000{,}000.
\]
The theorem gives
\[
    \mathbb{E}[D]
    \leq
    \max\left\{
        0,
        \left\lceil
        \log_{2000}\left(\frac{nM}{T}\right)
        \right\rceil
    \right\}
    +
    \frac{2000}{1999},
\]
where
\[
    M = (b-a)\|f\|_{\infty}.
\]
Since \(n/T = 100{,}000\), this becomes
\[
    \mathbb{E}[D]
    \leq
    \left\lceil
        \log_{2000}(100{,}000M)
    \right\rceil
    +
    \frac{2000}{1999}.
\]

For a uniform distribution on \([0,10000]\), we have
\[
    f(x)=\frac{1}{10000},
\]
so
\[
    M = 10000 \cdot \frac{1}{10000} = 1.
\]
Therefore,
\[
    \mathbb{E}[D]
    \leq
    \left\lceil
        \log_{2000}(100{,}000)
    \right\rceil
    +
    \frac{2000}{1999}
    =
    2 + \frac{2000}{1999}
    \approx 3.00.
\]
Thus, for a uniform input of one million elements, the theorem predicts that
only a small number of recursive levels are needed on average.

Next consider a normal distribution centered at \(5000\), with most values
lying within roughly \(500\) units of the mean. For example, we may take
\[
    X \sim \mathcal{N}(5000,250^2),
\]
so that about \(95\%\) of the mass lies within \(500\) units of the mean. The
maximum density occurs at the mean, so
\[
    \|f\|_{\infty}
    =
    \frac{1}{250\sqrt{2\pi}}.
\]
Hence
\[
    M
    =
    10000 \cdot \frac{1}{250\sqrt{2\pi}}
    \approx
    15.96.
\]
The theorem gives
\[
    \mathbb{E}[D]
    \leq
    \left\lceil
        \log_{2000}(100{,}000\cdot 15.96)
    \right\rceil
    +
    \frac{2000}{1999}
    =
    2 + \frac{2000}{1999}
    \approx 3.00.
\]
Even though the distribution is much more concentrated than the uniform
distribution, the expected depth bound remains very small.

As another example, consider an exponential distribution on \([0,10000]\). If
we use a truncated exponential distribution with rate
\[
    \lambda = \frac{1}{1000},
\]
then its density is
\[
    f(x)
    =
    \frac{\lambda e^{-\lambda x}}
    {1-e^{-10000\lambda}},
    \qquad 0\leq x\leq 10000.
\]
The maximum occurs at \(x=0\), so
\[
    \|f\|_{\infty}
    =
    \frac{\lambda}{1-e^{-10000\lambda}}.
\]
Therefore,
\[
    M
    =
    10000\cdot
    \frac{1/1000}{1-e^{-10}}
    \approx
    10.00.
\]
The theorem gives
\[
    \mathbb{E}[D]
    \leq
    \left\lceil
        \log_{2000}(100{,}000\cdot 10.00)
    \right\rceil
    +
    \frac{2000}{1999}
    =
    2 + \frac{2000}{1999}
    \approx 3.00.
\]

We also consider a more concentrated mixture distribution. This example is
useful because it is less favorable to distribution sorting than the previous
smooth distributions: most of the input lies in a narrow interval, while a small
fraction of values are spread across the full range.

Suppose that each element is drawn from the following mixture:
\[
    X \sim
    \begin{cases}
        \operatorname{Unif}([4999,5001]), & \text{with probability } 0.99,\\
        \operatorname{Unif}([0,10000]), & \text{with probability } 0.01.
    \end{cases}
\]
The density is therefore
\[
    f(x)
    =
    0.99\cdot \frac{1}{2}\mathbf{1}_{[4999,5001]}(x)
    +
    0.01\cdot \frac{1}{10000}\mathbf{1}_{[0,10000]}(x).
\]
The maximum density occurs inside the narrow interval \([4999,5001]\), where
both mixture components contribute. Hence
\[
    \|f\|_{\infty}
    =
    \frac{0.99}{2}
    +
    \frac{0.01}{10000}
    =
    0.495001.
\]
Thus
\[
    M
    =
    10000\|f\|_{\infty}
    =
    4950.01.
\]
The theorem gives
\[
    \mathbb{E}[D]
    \leq
    \left\lceil
        \log_{2000}(100{,}000\cdot 4950.01)
    \right\rceil
    +
    \frac{2000}{1999}
    =
    3 + \frac{2000}{1999}
    \approx 4.00.
\]

This is larger than in the uniform, normal, and exponential examples,
but still quite small. The increase is expected: the high-density interval contains many more elements
per unit length, so the cardinalities of buckets intersecting that interval
shrink more slowly across recursive levels. Nevertheless, because each recursive step divides
the current value interval into \(2000\) buckets, even a strongly concentrated
mixture requires only a few recursive levels in expectation.

\begin{table*}[t]
\centering
\caption{Expected recursion-depth bounds for representative distributions.}
\label{tab:expected-depth-examples}
\begin{tabular}{lcc}
\toprule
Distribution & \(M=(b-a)\|f\|_{\infty}\) & Bound on \(\mathbb{E}[D]\) \\
\midrule
Uniform on \([0,10000]\) & \(1\) & \(\approx 3.00\) \\
Normal, \(\mathcal{N}(5000,250^2)\) & \(15.96\) & \(\approx 3.00\) \\
Truncated exponential, \(\lambda=1/1000\) & \(10.00\) & \(\approx 3.00\) \\
Mixture distribution & \(4950.01\) & \(\approx 4.00\) \\
\bottomrule
\end{tabular}
\end{table*}

The results are summarized in Table~\ref{tab:expected-depth-examples}.
These examples show that, for broad classes of common distributions, the
expected recursion depth can remain very small even for large inputs. The
reason is that each level divides the value range into a large number of
buckets. Unless the input density is extremely concentrated, the number of
elements in the bucket containing a typical element decreases rapidly with
depth. This helps explain why recursive distribution sorting can behave close
to linearly in practice on many numeric inputs.

\subsection{Adversarial Inputs}
\label{sec:adversarial-inputs}

The inputs most likely to be adversarial for NFS Sort are those that cause
equal-width bucket refinement to make little progress. Since buckets are chosen
according to value intervals rather than ranks, the algorithm performs best
when the input values are reasonably spread across the current range. It can be
challenged when many elements occupy a very small portion of that range.

However, these cases generally need to be quite extreme. As the examples in
Section~\ref{sec:expected-depth} show, even substantially non-uniform
distributions can still lead to very small expected recursion depths. For
instance, the clustered mixture distribution summarized in
Table~\ref{tab:expected-depth-examples} places most of its probability mass in a
narrow interval, yet the expected-depth bound remains small because each
recursive step divides the value range into many buckets.

A more adversarial pattern consists of a large dense cluster together with a
small number of extreme outliers. For example, an array in which almost all
values lie in a very narrow interval, but a few values are near the minimum and
maximum of the numeric range, can cause the initial bucket widths to be much
larger than the dense region. As a result, almost all elements may be assigned
to a single bucket, producing little reduction in subproblem size.

Another difficult case is a recursively clustered distribution. In such an
input, the largest bucket at one level again contains values concentrated in a
small subinterval of its analytical bounds. This can cause several consecutive
recursive calls to isolate only a few elements while leaving most of the array
inside one dominant bucket.

Inputs with many nearly equal floating-point values can have a similar effect,
especially when a few distant values enlarge the observed range. Although the
values are not identical, their numerical spread may be small compared with the
current interval, so they repeatedly fall into the same bucket.

These adversarial inputs are not necessarily ordered in a special way in memory.
Their difficulty comes primarily from the geometry of the value distribution:
large gaps, extreme outliers, very dense clusters, or nested clusters can make
equal-width value partitioning poorly balanced. In practice, this means that an
input must be highly concentrated relative to its surrounding value range, often
over multiple recursive levels, before it becomes genuinely adversarial. Mild or
even high skew is not enough by itself; the problematic cases are extreme
distributions in which almost all elements remain concentrated inside a tiny
part of the current value interval over multiple recursive levels.

\section{Experiments}\label{section:experiments}
In this section, we test NFSS against several popular state-of-the-art sorting algorithms on both synthetic and real-world datasets.
For contending algorithms, we selected both a learned sorting algorithm, \emph{Balanced Learned Sort} (BLS)~\cite{ferragina2024balancedlearnedsortnew}, and plain, non-learned sorting algorithms: \emph{Boost Spreadsort}, \emph{In-place Parallel Super Scalar Samplesort} (IPS$^4$o)~\cite{axtmann2017ips4o}, \emph{Pattern-Defeating Quicksort} (PDQSort)~\cite{peters2021patterndefeatingquicksort}, and the standard C++ \texttt{std::sort}~\cite{musser1997introspective}.

Synthetic datasets were generated with the following distributions: \textbf{uniform}, \textbf{already sorted}, \textbf{reverse sorted}, \textbf{organ pipe}, \textbf{normal}, \textbf{nearly sorted}, \textbf{few unique values}, \textbf{exponential}, and \textbf{clustered} distributions.
For real-world datasets, we selected \emph{Search on Sorted Data} (SOSD) datasets~\cite{Marcus_2020,kipf2019sosdbenchmarklearnedindexes}. We used SOSD \textbf{books}, \textbf{Facebook}, \textbf{OpenStreetMap cell-id}, and \textbf{Wikipedia timestamp} datasets. We also include SOSD's distributional datasets, namely \textbf{lognormal}, \textbf{normal}, \textbf{dense-uniform}, and \textbf{sparse-uniform} keys, to cover controlled non-real-world key distributions.

Tests are run for various dataset sizes, ranging from $10^4$ to $10^9$ for all distributions. For this purpose, we sample the real-world datasets by shuffling them and taking the first $n$ items, where $n$ is the selected dataset size. This allows us to test the real-world datasets at arbitrary dataset sizes while preserving the underlying distribution.

All experiments were run on a single core of an AMD EPYC 9354 processor. The benchmark was compiled as an optimized C++20 binary using release settings \texttt{-O3 -march=native -DNDEBUG}. Each reported runtime is based on 10 runs. Each run uses a separate seed, which controls synthetic dataset generation and real-world dataset shuffling. The seed is shared across all algorithms, meaning each algorithm receives the same dataset as the others in each run. Timing is measured only around the sorting call using \texttt{std::chrono::steady\_clock}, and sorting correctness is validated.

\begin{figure*}[ht]
    \centering
    \includegraphics[width=\textwidth]{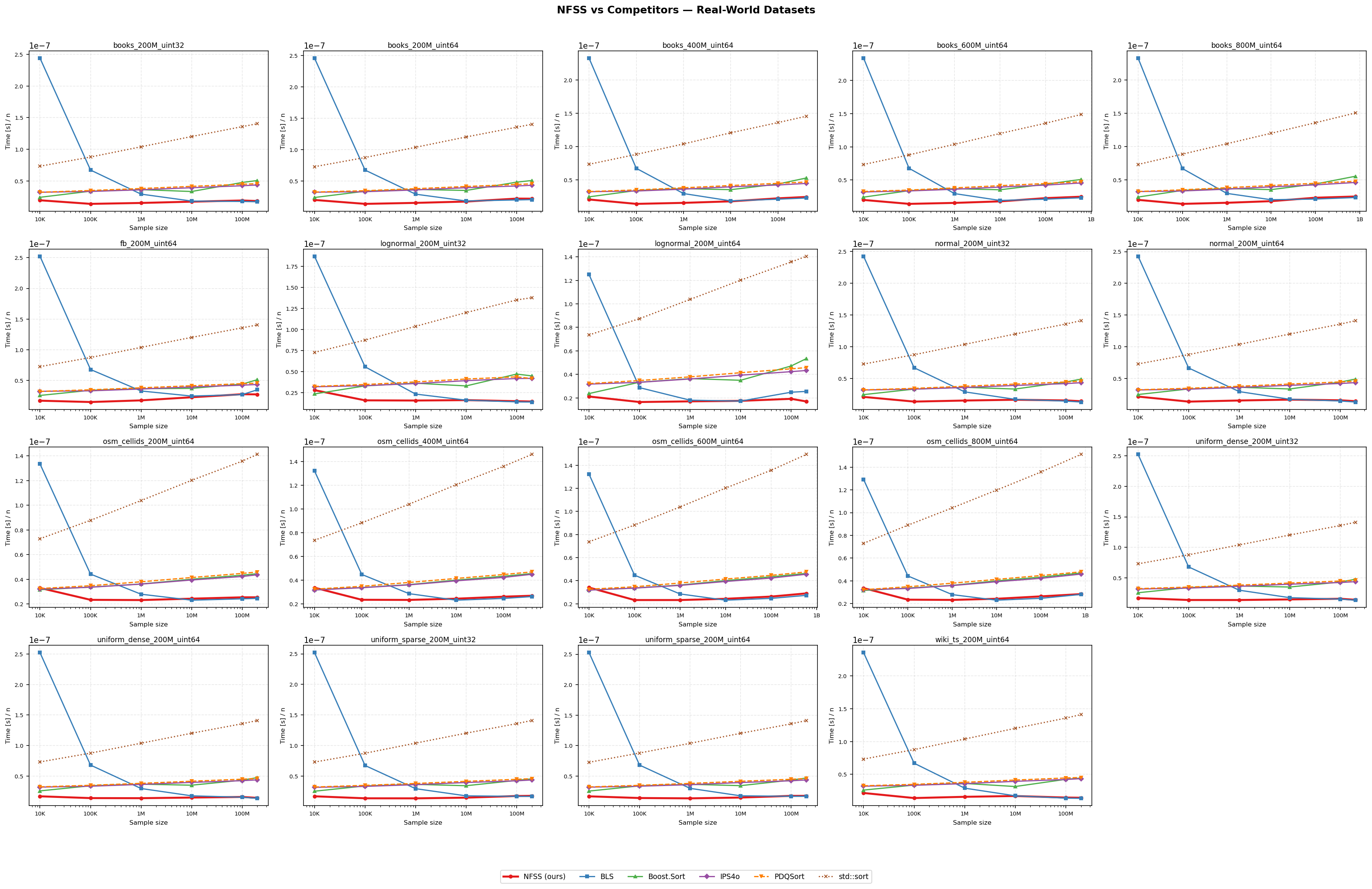}
    \caption{Time results of NFSS against other sorting algorithms on real-world datasets with varying dataset sizes.}
    \label{fig:real-world-time}
\end{figure*}

\begin{figure*}[h]
    \centering
    \includegraphics[width=0.5\linewidth]{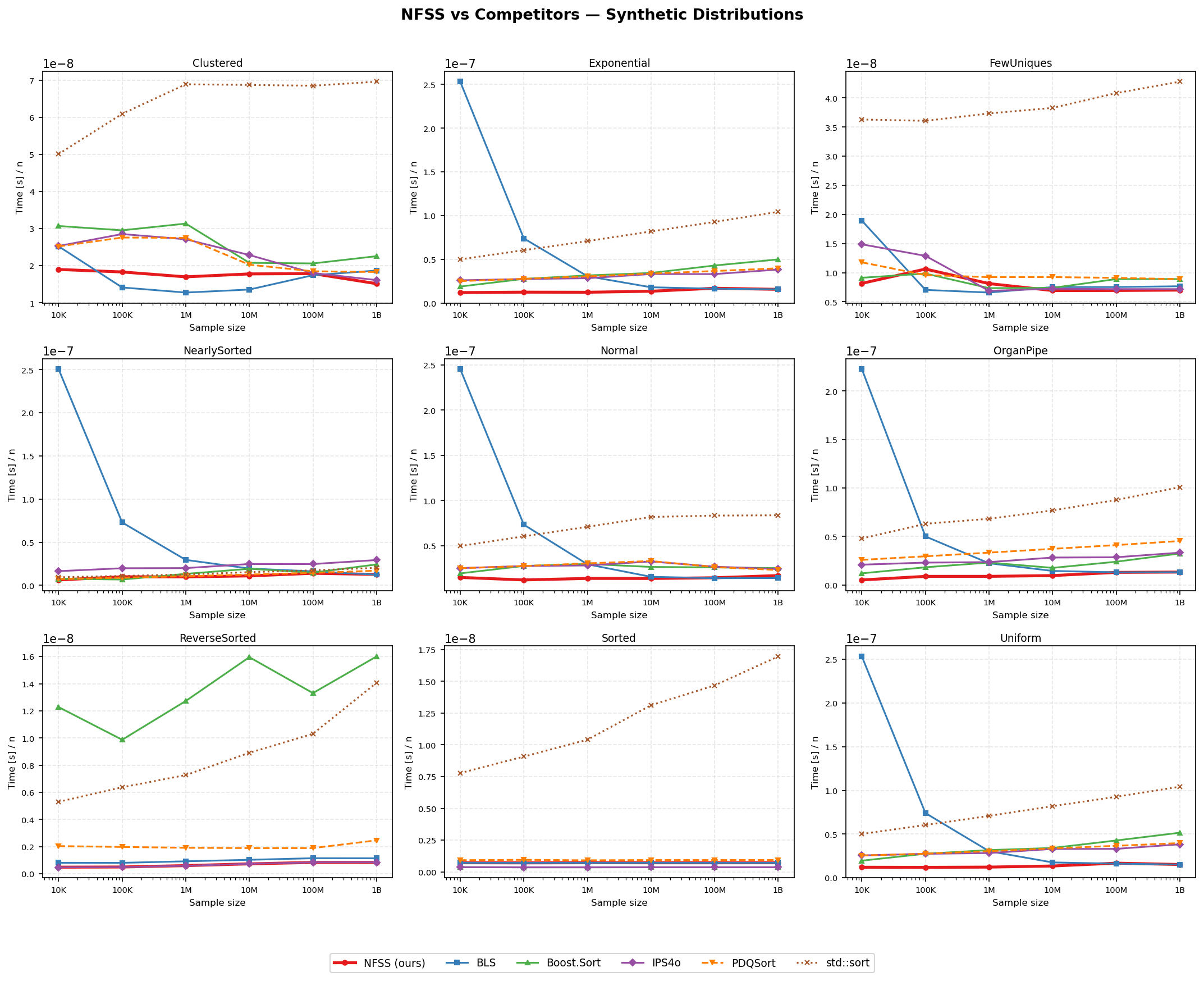}
    \caption{Time results of NFSS against other sorting algorithms on synthetic datasets with varying dataset sizes.}
    \label{fig:synth-time}
\end{figure*}

Notably, as shown in Figures \ref{fig:real-world-time} and \ref{fig:synth-time}, NFSS performs better than BLS on small dataset sizes, while sustaining comparable performance on larger sets and outperforming other sorting algorythms. Figures \ref{fig:real-speedup} and \ref{fig:synth-speedup} further visualise the average speed-up against other algorithms for every dataset. We base the speed-up ratios on a mean of all sizes.

\begin{figure*}[]
    \centering
    \includegraphics[width=0.7\linewidth]{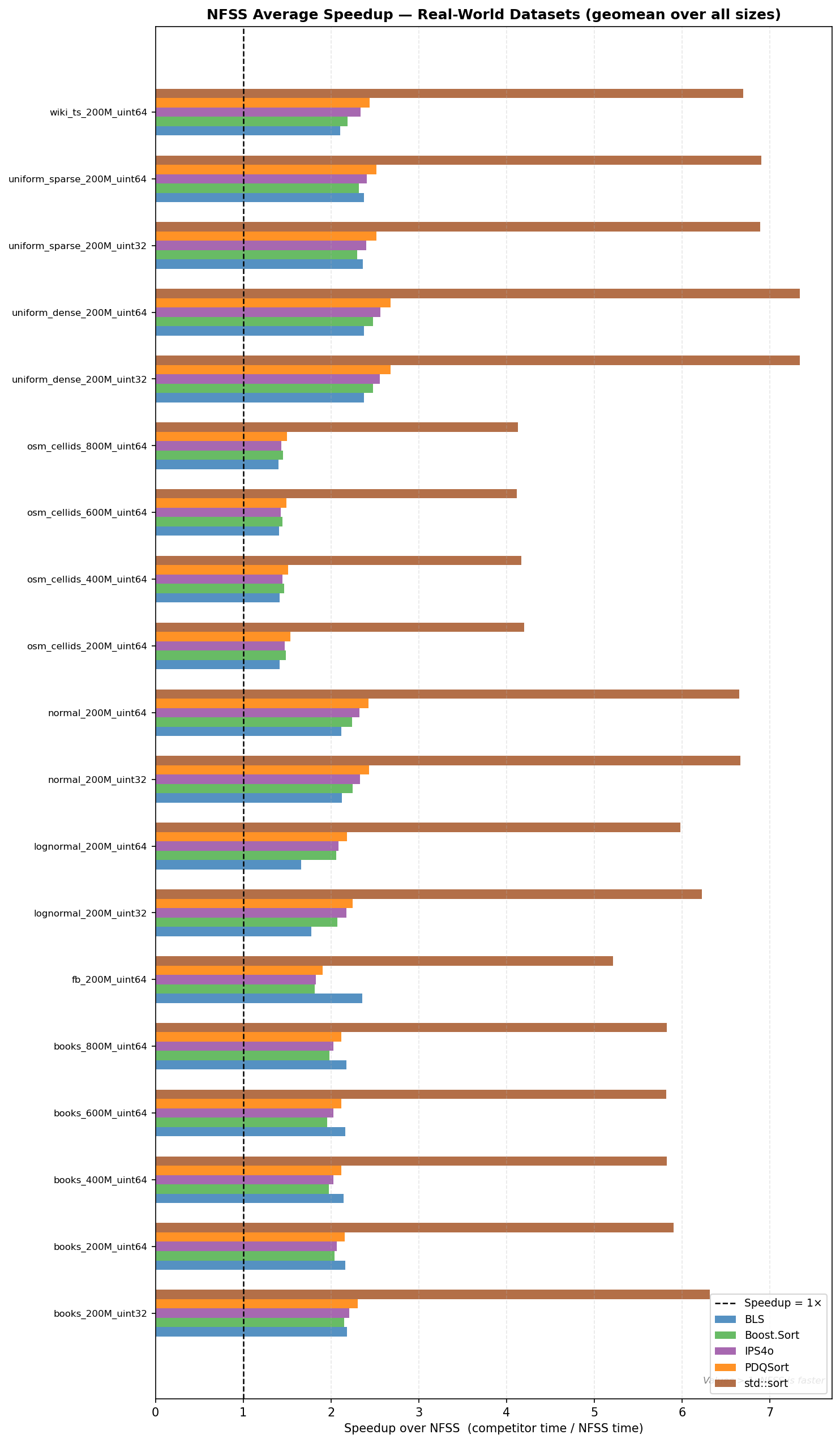}
    \caption{Average speedup over competitors on real-world datasets, expressed as competitor time/NFSS time, values larger than 1.0 mean NFSS is faster.}
    \label{fig:real-speedup}
\end{figure*}

\begin{figure*}[]
    \centering
    \includegraphics[width=0.5\linewidth]{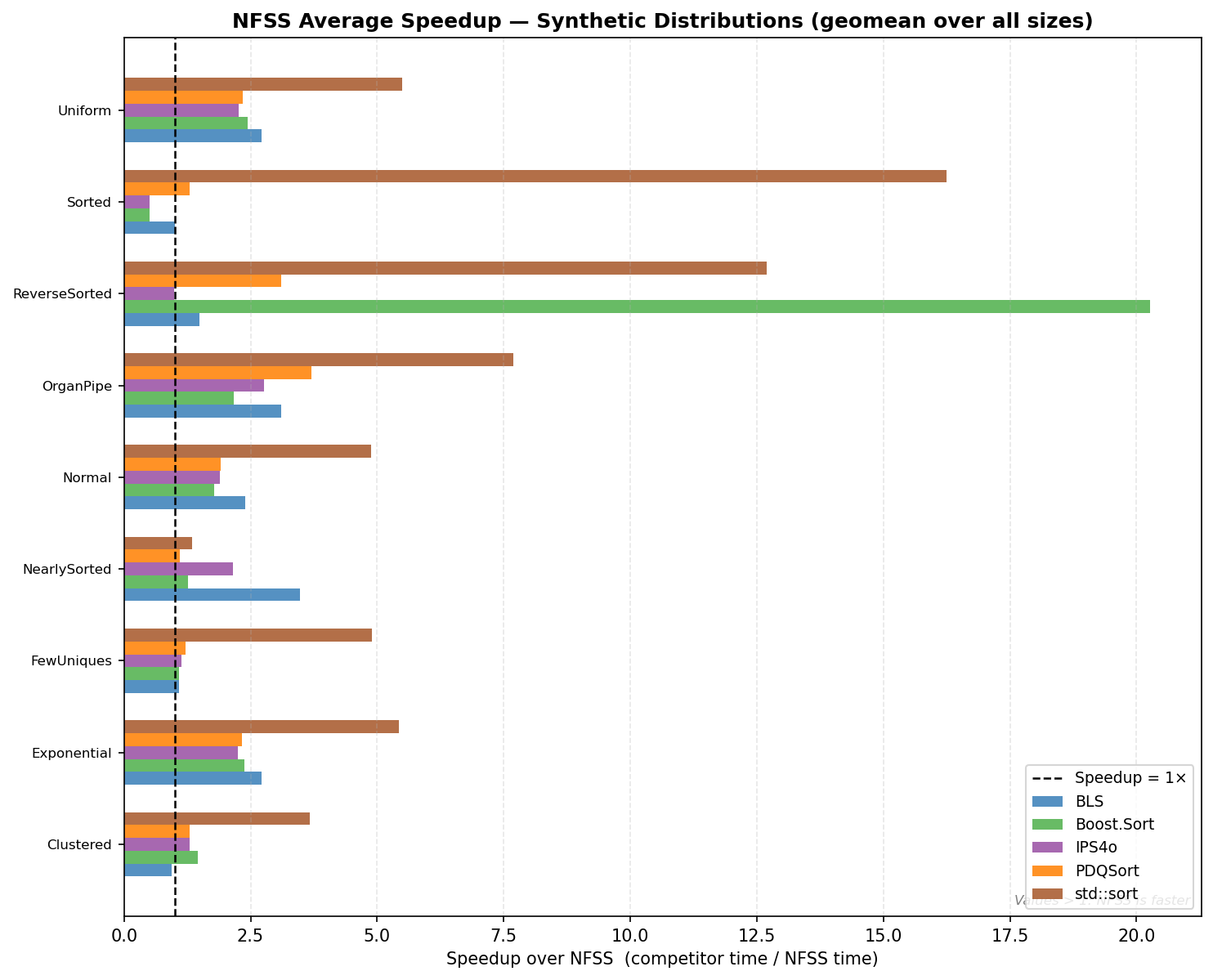}
    \caption{Average speedup over competitors on synthetic datasets, expressed as competitor time/NFSS time, values larger than 1.0 mean NFSS is faster.}
    \label{fig:synth-speedup}
\end{figure*}
\section{Conclusion}

In this paper, we introduced \emph{Need for Speed Sort} (NFS Sort), a recursive distribution-based sorting algorithm designed to efficiently handle numeric arrays. By systematically refining dense value intervals and propagating analytical bounds, NFS Sort exploits the structure of the input's value range while circumventing the overhead of repeated linear scans. Furthermore, the integration of a strategic fallback mechanism guarantees a worst-case time complexity of $\mathcal{O}(n \log n)$ and an auxiliary space footprint of $\mathcal{O}(\log n)$, ensuring robust performance even against highly adversarial inputs.

Our theoretical analysis demonstrated that, for a wide variety of input distributions, the expected recursion depth remains shallow, allowing the algorithm to operate in near-linear time in practice. This theoretical efficiency is strongly supported by our extensive empirical evaluation across diverse synthetic and real-world (SOSD) datasets. As shown in Section~\ref{section:experiments}, NFS Sort is highly competitive with established state-of-the-art sorting algorithms. Most notably, it consistently outperforms recent learned sorting techniques, such as Balanced Learned Sort (BLS), particularly on smaller datasets, while maintaining comparable and highly efficient performance on larger scales. Against other non-learned state-of-the-art sorts, NFS Sort is competitive or better across all scales.

Ultimately, NFS Sort successfully combines the theoretical advantages of distribution sorting with modern memory-management techniques to meet the practical demands of contemporary software systems. Future work could explore extending the algorithm's distribution logic to non-numeric keys, parallelizing the recursive scatter steps to leverage multi-core architectures, and further fine-tuning the fallback heuristics for specialized hardware environments.
\appendix

\section{Naive NFS Sort Pseudocode}
\label{app:naive-nfs}

\begin{figure*}[t]
\begin{lstlisting}[caption={Naive NFS Sort}, label={alg:naive-nfs}]
void nfss_sort(double[] array) {

    // find minimum and maximum values in the array
    double min, max;
    (min, max) = array.getMinAndMax();

    // initialize helper array for scatter step
    double[] helper_array = new double[C];

    nfss_core(array, helper_array, array.start, array.end, min, max, 0);

    insertion_sort(array);

}

void nfss_core(
    double[] array,
    double[] helper_array,
    int array_start,
    int array_end,
    double min,
    double max,
    int pathological_recursion_counter
) {
    // calculate array basic information
    int array_size = array_end - array_start;
    double interval = max - min;
    double inverse_interval = 1.0 / interval;

    int bucket_number;
    int[] bucket_offsets;

    // Decide on strategy and bucket number based on array size
    if (array_size > 5000000) {
        bucket_number = 2000;
        bucket_offsets = IPSo_inspired_fragment_scatter(array, helper_array, array_start, array_end, min, max, bucket_number, inverse_interval);
    } else if (array_size > 1000000) {
        bucket_number = 1000;
        bucket_offsets = IPSo_inspired_fragment_scatter(array, helper_array, array_start, array_end, min, max, bucket_number, inverse_interval);
    } else if (array_size > 500000) {
        bucket_number = 200;
        bucket_offsets = IPSo_inspired_fragment_scatter(array, helper_array, array_start, array_end, min, max, bucket_number, inverse_interval);
    } else if (array_size > 10000) {
        bucket_number = 100;
        bucket_offsets = IPSo_inspired_fragment_scatter(array, helper_array, array_start, array_end, min, max, bucket_number, inverse_interval);
    } else {
        bucket_number = minimum(6500, array_size);
        bucket_offsets = naive_scatter(array, helper_array, array_start, array_end, min, max, bucket_number, inverse_interval);
    }

    // calculate interval size and "bad" critical mass
    double interval_size = interval / (bucket_number - 1);
    double critical_mass = $\alpha$ * array_size;

    // decide what to do with each bucket based on bucket size
    for (int i = 0; i < bucket_number; i++) {
        int bucket_start = array_start + bucket_offsets[i];
        int bucket_end = array_start + bucket_offsets[i + 1];
        int bucket_size = bucket_end - bucket_start;

        // insertion sort at the end will handle these buckets
        if (bucket_size <= small_threshold) continue;

        // pathological buckets that have hardly lost any mass
        if (bucket_size > critical_mass) {
        
            // we try a fixed number of times with recursion,
            // and if it doesn't work we fall back on some other efficient sort
            if (pathological_recursion_counter < 1) {
                nfss_core(array, helper_array, bucket_start, bucket_end,
                    min + interval_size * i,
                    min + interval_size * (i + 1),
                    pathological_recursion_counter + 1);
            } else {
                fallback_sort(array, bucket_start, bucket_end);
            }
        } else {
            nfss_core(array, helper_array, bucket_start, bucket_end,
                min + interval_size * i,
                min + interval_size * (i + 1),
                0);
        }
    }

}
\end{lstlisting}
\end{figure*}

\bibliographystyle{ACM-Reference-Format}
\bibliography{references}

\end{document}